\documentclass[letterpaper, 10 pt, conference]{ieeeconf}  

\IEEEoverridecommandlockouts                              
\usepackage[utf8]{inputenc}
\usepackage{color}
\usepackage{array}

\usepackage{float}
\usepackage{amsmath}
\usepackage{amssymb}
\usepackage{graphicx}
\usepackage{longtable}
\usepackage{multirow}

\usepackage{verbatim}
\usepackage{amsmath}
\usepackage{amssymb}
\usepackage{graphicx}

\usepackage[unicode=true,
bookmarks=false,
breaklinks=false,pdfborder={0 0 1},colorlinks=false]
{hyperref}
\hypersetup{
	colorlinks,bookmarksopen=red,bookmarksnumbered=red,citecolor=blue,urlcolor=red}
\usepackage{cite}

\floatstyle{ruled}
\newfloat{algorithm}{tbp}{loa}
\providecommand{\algorithmname}{Algorithm}
\floatname{algorithm}{\protect\algorithmname}

\newtheorem{lem}{Lemma}

\newtheorem{thm}{Theorem}

\newtheorem{assum}{Assumption}
\usepackage{cite}

\title{\LARGE \bf Nonlinear Deterministic Filter for Inertial Navigation and Bias Estimation with Guaranteed Performance}

\author{Ajay Singh Ludher, Marium Tawhid, and Hashim A. Hashim\\
	Software Engineering, Department of Engineering and Applied Science\\
	Thompson Rivers University,	Kamloops, British Columbia, Canada, V2C-0C8\\
	ludhera17@mytru.ca, tawhidm16@mytru.ca, and hhashim@tru.ca
	\thanks{This work was supported in part by Thompson Rivers University Internal research fund, RGS-2020/21 IRF, \# 102315.}}

\begin{document}
\bstctlcite{IEEEexample:BSTcontrol}

\maketitle
\thispagestyle{empty}
\pagestyle{empty}

\begin{abstract}
Unmanned vehicle navigation concerns estimating attitude, position,
and linear velocity of the vehicle the six degrees of freedom (6 DoF).
It has been known that the true navigation dynamics are highly nonlinear
modeled on the Lie Group of $\mathbb{SE}_{2}(3)$. In this paper,
a nonlinear filter for inertial navigation is proposed. The filter
ensures systematic convergence of the error components starting from
almost any initial condition. Also, the errors converge asymptotically
to the origin. Experimental results validates the robustness of the
proposed filter.
\end{abstract}

\section{Introduction}

Navigation solutions are key element in autonomous vehicles \cite{Hashim2021AESCTE,mourikis2007multi}.
Navigation estimation algorithms considers the vehicle orientation
(attitude), position, and linear velocity to be completely unknown
and require estimating the aforementioned three elements\cite{hashim2021ACC,hashim2021_COMP_ENG_PRAC}.
Navigation solutions become indispensable if global positioning systems
(GPS) are unreliable. A set of measurements is required for the estimation
process. The vehicle's orientation can be estimated using for instance,
inertial measurement unit (IMU) \cite{hashim2019SO3Wiley,hashim2018SO3Stochastic,odry2021open,jensen2011generalized,odry2018kalman},
while vehicle's pose (orientation and position) can be estimated using
IMU and vision unit \cite{hashim2020SE3Stochastic}. Recently, other
potential solutions emerged to estimate the vehicle's pose as well
as map the unknown environment such as nonlinear deterministic filter
for simultaneous localization and mapping (SLAM) \cite{hashim2020TITS_SLAM,hashim2021T_SMCS_SLAM}
and nonlinear stochastic filter for SLAM \cite{Hashim2021AESCTE}.
The family of pose and SLAM filters requires linear velocity to be
known. In practice, linear velocity is hard to obtain in GPS-denied
regions and it is challenging to reconstruct.

Traditionally, inertial navigation used to addressed using Gaussian
filters, such as Kalman filter \cite{zhao2016analysis}, extended
Kalman filter \cite{hsu2017wearable}, unscented Kalman filter \cite{allotta2016unscented},
and particle filter which is non-Gaussian filter \cite{blok2019robot}
among others. Gaussian filters are based on linear approximation while
particle filters do not have clear measure for optimal performance.
However, navigation dynamics of a vehicle navigating in three dimensional
(3D) space are highly nonlinear. Also, the dynamics cannot be classified
as right or left invariant. Thereby, a recent development of filters
for inertial navigation on the Lie Group have been developed, for
instance invariant extended Kalman filter (IEKF) on the Lie Group
of $\mathbb{SE}_{2}(3)$ \cite{barrau2016invariant}, a Riccati filter
design \cite{hua2018riccati}, and a nonlinear stochastic filter on
the Lie group of $\mathbb{SE}_{2}(3)=\mathbb{SO}\left(3\right)\times\mathbb{R}^{3}\times\mathbb{R}^{3}\subset\mathbb{R}^{5\times5}$
\cite{hashim2021ACC,hashim2021_COMP_ENG_PRAC}. The filters in \cite{barrau2016invariant,hua2018riccati,hashim2021ACC,hashim2021_COMP_ENG_PRAC}
are developed directly on $\mathbb{SE}_{2}(3)$, however measures
of transient and steady-state can be defined.

To sum up, the navigation dynamics are highly nonlinear posed on the
Lie group of $\mathbb{SE}_{2}(3)$. Also, the transient and steady-state
error has to be taken into account. It is worth noting that navigation
filters relies on gyroscope and accelerometer measurements, and therefore,
uncertainties in gyroscope and accelerometer measurements should be
addressed. Hence, this paper consider the previously mentioned challenges
through the following set of contributions: (1) A nonlinear filter
on the Lie Group of $\mathbb{SE}_{2}(3)$ for inertial navigation
with predefined measures of transient and steady-state performance
is proposed; (2) the closed loop errors are shown to be almost globally
asymptotically stable, and when provided with data from low-cost IMU
and feature sensors; and (3) the presence of unknown bias in IMU measurements
is successfully tackled. 

The rest of the paper is organized as follows: Section \ref{sec:SE3_Problem-Formulation}
presents important notation and identities, and introduces the true
navigation problem. Section \ref{sec:SLAM_Filter} present the concept
of guaranteed measures of transient and steady-state performance and
proposes a novel nonlinear filter. Section \ref{sec:SE3_Simulations}
shows experimental results. Finally, Section \ref{sec:SE3_Conclusion}
concludes the work.

\section{Navigation Framework\label{sec:SE3_Problem-Formulation}}

\subsection{Preliminaries}

Let $\mathbb{R}$ set of real numbers, and $\mathbb{R}^{n\times m}$
denote denote an $n$-by-$m$ real real numbers. $\mathbf{I}_{n}$
is an $n$-by-$n$ identity matrix and $0_{n\times m}$ is an $n$-by-$m$
dimensional matrix of zeros. $\left\{ \mathcal{I}\right\} $ stands
for the fixed inertial-frame and $\left\{ \mathcal{B}\right\} $ corresponds
to the fixed body-frame attached to a vehicle moving in 3D space.
The vehicle's orientation in 3D space, known as attitude, is described
by $R\in\mathbb{SO}\left(3\right)$ where $\mathbb{SO}\left(3\right)$
is a short-hand notation for Special Orthogonal Group given by
\[
\mathbb{SO}\left(3\right)=\left\{ \left.R\in\mathbb{R}^{3\times3}\right|RR^{\top}=R^{\top}R=\mathbf{I}_{3}\text{, }{\rm det}\left(R\right)=+1\right\} 
\]
where ${\rm det}\left(\cdot\right)$ stands for a determinant. $\mathfrak{so}\left(3\right)$
is the Lie-algebra of $\mathbb{SO}\left(3\right)$ denoted by
\begin{align*}
\mathfrak{so}\left(3\right) & =\left\{ \left.\left[x\right]_{\times}\in\mathbb{R}^{3\times3}\right|\left[x\right]_{\times}^{\top}=-\left[x\right]_{\times},x\in\mathbb{R}^{3}\right\} 
\end{align*}
\[
\left[x\right]_{\times}=\left[\begin{array}{ccc}
0 & -x_{3} & x_{2}\\
x_{3} & 0 & -x_{1}\\
-x_{2} & x_{1} & 0
\end{array}\right]\in\mathfrak{so}\left(3\right),\hspace{1em}x=\left[\begin{array}{c}
x_{1}\\
x_{2}\\
x_{3}
\end{array}\right]
\]
The following mappings are defined
\begin{align}
\mathbf{vex}([x]_{\times}) & =x,\hspace{1em}\forall x\in\mathbb{R}^{3}\label{eq:NAV_VEX}\\
\boldsymbol{\mathcal{P}}_{a}(M) & =\frac{1}{2}(M-M^{\top})\in\mathfrak{so}\left(3\right),\hspace{1em}\forall M\in\mathbb{R}^{3\times3}\label{eq:NAV_Pa}\\
\boldsymbol{\Upsilon}\left(M\right) & =\mathbf{vex}(\boldsymbol{\mathcal{P}}_{a}(M))\in\mathbb{R}^{3},\hspace{1em}\forall M\in\mathbb{R}^{3\times3}\label{eq:NAV_VEX_a}
\end{align}
Define $||R||_{{\rm I}}$ as the Euclidean distance of $R\in\mathbb{SO}\left(3\right)$
such that
\begin{equation}
||R||_{{\rm I}}=\frac{1}{4}{\rm Tr}\{\mathbf{I}_{3}-R\}\in\left[0,1\right]\label{eq:NAV_Ecul_Dist}
\end{equation}
For more information about attitude and pose notation visit \cite{hashim2019SO3Wiley,hashim2020SE3Stochastic}.
The extended form of the Special Euclidean Group denoted by $\mathbb{SE}_{2}\left(3\right)=\mathbb{SO}\left(3\right)\times\mathbb{R}^{3}\times\mathbb{R}^{3}\subset\mathbb{R}^{5\times5}$
is defined by

\begin{align}
\mathbb{SE}_{2}\left(3\right)= & \{\left.X\in\mathbb{R}^{5\times5}\right|R\in\mathbb{SO}\left(3\right),P,V\in\mathbb{R}^{3}\}\label{eq:NAV_SE2_3}\\
X= & \left[\begin{array}{ccc}
R & P & V\\
0_{1\times3} & 1 & 0\\
0_{1\times3} & 0 & 1
\end{array}\right]\in\mathbb{SE}_{2}\left(3\right)\label{eq:NAV_X}
\end{align}
where $X$, $R$, $P$, and $V$ denote a homogeneous navigation matrix,
rigid-body's orientation, position, and linear velocity, respectively,
Let $\mathcal{U}_{\mathcal{M}}=\mathfrak{so}\left(3\right)\times\mathbb{R}^{3}\times\mathbb{R}^{3}\times\mathbb{R}$
be a submanifold of $\mathbb{R}^{5\times5}$ where
\begin{align}
\mathcal{U}_{\mathcal{M}} & =\left\{ \left.u\left(\text{\ensuremath{\left[\Omega\right]_{\times}}},V,a,\kappa\right)\right|\text{\ensuremath{\left[\Omega\right]_{\times}}}\in\mathfrak{so}\left(3\right),V,a\in\mathbb{R}^{3},\kappa\in\mathbb{R}\right\} \nonumber \\
& u\left(\text{\ensuremath{\left[\Omega\right]_{\times}}},V,a,\kappa\right)=\left[\begin{array}{ccc}
\left[\Omega\right]_{\times} & V & a\\
0_{1\times3} & 0 & 0\\
0_{1\times3} & \kappa & 0
\end{array}\right]\in\mathcal{U}_{\mathcal{M}}\label{eq:NAV_u}
\end{align}
with $\Omega\in\mathbb{R}^{3}$ and $a\in\mathbb{R}^{3}$ being the
vehicle's true angular velocity and the apparent acceleration composed
of all non-gravitational forces affecting the vehicle. It should be
remarked that $R,\Omega,a\in\left\{ \mathcal{B}\right\} $ and $P,V\in\left\{ \mathcal{I}\right\} $.

\subsection{Navigation Dynamics and Measurements}

The true 3D navigation dynamics are highly nonlinear described by
\cite{hashim2021ACC,hashim2021_COMP_ENG_PRAC}
\begin{equation}
\begin{cases}
\dot{R} & =R\left[\Omega\right]_{\times}\\
\dot{P} & =V\\
\dot{V} & =Ra+\overrightarrow{\mathtt{g}}
\end{cases},\hspace{1em}\underbrace{\dot{X}=XU-\mathcal{\mathcal{G}}X}_{\text{Compact form}}\label{eq:NAV_True_dot}
\end{equation}
with the left portion of \eqref{eq:NAV_True_dot} being the detailed
navigation dynamics and the notation were defined in \eqref{eq:NAV_X}
and \eqref{eq:NAV_u}. Also, $\overrightarrow{\mathtt{g}}$ refers
to a gravity vector. Note that $\dot{X}:\mathbb{SE}_{2}\left(3\right)\times\mathbb{R}^{3}\times\mathbb{R}^{3}\rightarrow T_{X}\mathbb{SE}_{2}\left(3\right)$
\cite{hashim2021ACC,hashim2021_COMP_ENG_PRAC}. The right portion
of \eqref{eq:NAV_True_dot} describes the compact form $\dot{X}:\mathbb{SE}_{2}\left(3\right)\times\mathcal{U}_{\mathcal{M}}\rightarrow T_{X}\mathbb{SE}_{2}\left(3\right)$
with $U=\underbrace{\left[\begin{array}{ccc}
	\left[\Omega\right]_{\times} & 0_{3\times1} & a\\
	0_{1\times3} & 0 & 0\\
	0_{1\times3} & 1 & 0
	\end{array}\right]}_{u(\text{\ensuremath{\left[\Omega\right]_{\times}}},0_{3\times1},a,1)}$, $\mathcal{\mathcal{G}}=\underbrace{\left[\begin{array}{ccc}
	0_{3\times3} & 0_{3\times1} & -\overrightarrow{\mathtt{g}}\\
	0_{1\times3} & 0 & 0\\
	0_{1\times3} & 1 & 0
	\end{array}\right]}_{u(0_{3\times3},0_{3\times1},-\overrightarrow{\mathtt{g}},1)}$ \cite{hashim2021ACC,hashim2021_COMP_ENG_PRAC}. The measurements of $\Omega$ and $a$ follows:
\begin{equation}
\begin{cases}
\Omega_{m} & =\Omega+b_{\Omega}\in\mathbb{R}^{3}\\
a_{m} & =a+b_{a}\in\mathbb{R}^{3}
\end{cases}\label{eq:NAV_XVelcoity}
\end{equation}
where $b_{\Omega}$ and $b_{a}$ denote unknown bias. Consider $n$
features in the environment where $p_{i}\in\mathbb{R}^{3}$ denotes
the $i$th feature position $p_{i}\in\left\{ \mathcal{I}\right\} $
for all $i=1,2,\ldots,n$. Let the $i$th feature measurement be
\begin{equation}
y_{i}=R^{\top}(p_{i}-P)+b_{i}^{y}+n_{i}^{y}\in\mathbb{R}^{3}\label{eq:NAV_Vec_Landmark}
\end{equation}
whee $b_{i}^{y}$ is an unknown bias and $n_{i}^{y}$ is noise. 

\begin{assum}\label{Assum:NAV_1Landmark}At least three non-collinear
	features available for measurement at each time instant. \end{assum}

\subsection{Navigation Matrix: Estimate, Error, and Measurements Setup\label{subsec:Navigation-Matrix}}

Let the estimate of the true homogeneous navigation matrix $X\in\mathbb{SE}_{2}\left(3\right)$
defined in \eqref{eq:NAV_X} be
\begin{equation}
\hat{X}=\left[\begin{array}{ccc}
\hat{R} & \hat{P} & \hat{V}\\
0_{1\times3} & 1 & 0\\
0_{1\times3} & 0 & 1
\end{array}\right]\in\mathbb{SE}_{2}\left(3\right)\label{eq:NAV_X_est}
\end{equation}
where $\hat{R}\in\mathbb{SO}\left(3\right)$, $\hat{P}\in\mathbb{R}^{3}$,
and $\hat{V}\in\mathbb{R}^{3}$ denote estimates of the true orientation,
position, and velocity, respectively. Define the error between $X$
and $\hat{X}$ as
\begin{align}
\tilde{X}=X\hat{X}^{-1} & =\left[\begin{array}{ccc}
\tilde{R} & \tilde{P} & \tilde{V}\\
0_{1\times3} & 1 & 0\\
0_{1\times3} & 0 & 1
\end{array}\right]\label{eq:NAV_X_error}
\end{align}
where $\tilde{R}=R\hat{R}^{\top}$, $\tilde{P}=P-\tilde{R}\hat{P}$,
and $\tilde{V}=V-\tilde{R}\hat{V}$. A set of measurements have to
be defined to be subsequently used as part of the filter design, part
of the following definitions can be found at \cite{hashim2020LetterSLAM,hashim2021ACC,hashim2021_COMP_ENG_PRAC}.
Let us begin by defining the error:
\begin{align}
\overset{\circ}{\tilde{y}}_{i} & =\overline{p}_{i}-\tilde{X}^{-1}\overline{p}_{i}=\overline{p}_{i}-\hat{X}\overline{y}_{i}\nonumber \\
& =\left[(p_{i}-\hat{R}y_{i}-\hat{P})^{\top},0,0\right]^{\top}\in\overset{\circ}{\mathcal{M}}\label{eq:NAV_y_error}
\end{align}
where $\overset{\circ}{\tilde{y}}_{i}=[\tilde{y}_{i}^{\top},0,0]^{\top}\in\overset{\circ}{\mathcal{M}}$,
$p_{i}-\hat{R}y_{i}-\hat{P}=\tilde{p}_{i}-\tilde{P}$, $\tilde{p}_{i}=\hat{p}_{i}-\tilde{R}p_{i}$,
and $\tilde{P}=P-\tilde{R}\hat{P}$. Define the following components:
$s_{T}=\sum_{i=1}^{n}s_{i}$, $p_{c}=\frac{1}{s_{T}}\sum_{i=1}^{n}s_{i}p_{i}$,
and $M=\sum_{i=1}^{n}s_{i}(p_{i}-p_{c})(p_{i}-P)^{\top}=\sum_{i=1}^{n}s_{i}p_{i}p_{i}^{\top}-2\sum_{i=1}^{n}s_{i}p_{i}p_{c}^{\top}+s_{T}p_{c}p_{c}^{\top}$
such that
\begin{align}
M & =\sum_{i=1}^{n}s_{i}p_{i}p_{i}^{\top}-s_{T}p_{c}p_{c}^{\top}\label{eq:NAV_M}
\end{align}
where $s_{i}>0$ denotes the sensor confidence level of the $i$th
landmark, and $n\geq3$ as defined in Assumption \ref{Assum:NAV_1Landmark}.
One finds $\sum_{i=1}^{n}s_{i}(p_{i}-p_{c})y_{i}^{\top}\hat{R}^{\top}=\sum_{i=1}^{n}s_{i}(p_{i}-p_{c})(p_{i}-P)^{\top}\tilde{R}$
which means that
\begin{align}
\sum_{i=1}^{n}s_{i}\left(p_{i}-p_{c}\right)y_{i}^{\top}\hat{R}^{\top}= & M\tilde{R}\label{eq:NAV_MR}
\end{align}
Additionally, the following result can be obtained $\sum_{i=1}^{n}s_{i}\tilde{y}_{i}=\sum_{i=1}^{n}s_{i}(p_{i}-\hat{R}y_{i}-\hat{P})=\sum_{i=1}^{n}s_{i}(p_{i}-\tilde{R}^{\top}p_{i})+\sum_{i=1}^{n}s_{i}\tilde{R}^{\top}\tilde{P}$
such that $\sum_{i=1}^{n}s_{i}\tilde{y}_{i}=s_{T}\tilde{R}^{\top}\tilde{P}_{\varepsilon}$
with $\tilde{P}_{\varepsilon}=\tilde{P}-(\mathbf{I}_{3}-\tilde{R})p_{c}$.
Note that $\tilde{R}\rightarrow\mathbf{I}_{3}$ indicates that $\tilde{P}_{\varepsilon}\rightarrow\tilde{P}$
and $\tilde{R}=\mathbf{I}_{3}$ implying that $\tilde{R}^{\top}\tilde{P}_{\varepsilon}=\tilde{P}_{\varepsilon}=\tilde{P}$.
Summing up the above derivations, the following set expressed in terms
of vector measurements will be used in the filter design \cite{hashim2021ACC,hashim2021_COMP_ENG_PRAC}:
\begin{equation}
\begin{cases}
p_{c} & =\frac{1}{s_{T}}\sum_{i=1}^{n}s_{i}p_{i},\hspace{1em}s_{T}=\sum_{i=1}^{n}s_{i}\\
M & =\sum_{i=1}^{n}s_{i}p_{i}p_{i}^{\top}-s_{T}p_{c}p_{c}^{\top}\\
\tilde{y}_{i} & =p_{i}-\hat{R}y_{i}-\hat{P}\\
M\tilde{R} & =\sum_{i=1}^{n}s_{i}\left(p_{i}-p_{c}\right)y_{i}^{\top}\hat{R}^{\top}\\
\tilde{R}^{\top}\tilde{P}_{\varepsilon} & =\frac{1}{s_{T}}\sum_{i=1}^{n}s_{i}\tilde{y}_{i}
\end{cases}\label{eq:NAV_Set_Measurements}
\end{equation}

\subsection{Nonlinear Filter Framework and Error Dynamics\label{subsec:General-Nonlinear-Observer}}

\begin{lem}
	\label{Lemm:SLAM_Lemma1}Let $\tilde{R}\in\mathbb{SO}\left(3\right)$,
	$M=M^{\top}\in\mathbb{R}^{3\times3}$. Define $\overline{\mathbf{M}}={\rm Tr}\{M\}\mathbf{I}_{3}-M$
	such that $\underline{\lambda}_{\overline{\mathbf{M}}}=\underline{\lambda}(\overline{\mathbf{M}})$
	and $\overline{\lambda}_{\overline{\mathbf{M}}}=\overline{\lambda}(\overline{\mathbf{M}})$
	are the minimum and the maximum eigenvalues of $\overline{\mathbf{M}}$,
	respectively. Then, one has
	\begin{align}
	\frac{\underline{\lambda}_{\overline{\mathbf{M}}}}{2}(1+{\rm Tr}\{\tilde{R}\})||M\tilde{R}||_{{\rm I}} & \leq||\boldsymbol{\Upsilon}(M\tilde{R})||^{2}\leq2\overline{\lambda}_{\overline{\mathbf{M}}}||M\tilde{R}||_{{\rm I}}\label{eq:SLAM_lemm1}
	\end{align}
	\begin{proof}See (\cite{hashim2019SO3Wiley}, Lemma 1).\end{proof}
\end{lem}
From Lemma \ref{Lemm:SLAM_Lemma1}, define $\overline{\mathbf{M}}={\rm Tr}\{M\}\mathbf{I}_{3}-M$
given that $\lambda(M)=\{\lambda_{1},\lambda_{2},\lambda_{3}\}$ with
$\lambda_{3}\geq\lambda_{2}\geq\lambda_{1}$. In view of Assumption
\eqref{Assum:NAV_1Landmark}, at least two of the eigenvalues in the
set $\lambda(M)$ are greater than zero and therefore $\lambda(\overline{\mathbf{M}})=\{\lambda_{3}+\lambda_{2},\lambda_{3}+\lambda_{1},\lambda_{2}+\lambda_{1}\}$
see Section 4.2 in \cite{hashim2019SO3Wiley}.

\section{Nonlinear Navigation Filter Design \label{sec:SLAM_Filter}}

\subsection{Systematic Convergence}

Consider the following error in view of the measurements in \eqref{eq:NAV_Set_Measurements}:
\begin{equation}
e=[e_{1},e_{2},e_{3},e_{4}]^{\top}=\left[||M\tilde{R}||_{{\rm I}},\tilde{P}_{\varepsilon}^{\top}\tilde{R}\right]^{\top}\in\mathbb{R}^{4}\label{eq:NAV_Vec_error}
\end{equation}
Define $\xi_{i}:\mathbb{R}_{+}\to\mathbb{R}_{+}$ as a positive smooth
and time-decreasing function \cite{bechlioulis2008robust}:
\begin{equation}
\xi_{i}\left(t\right)=\left(\xi_{i}^{0}-\xi_{i}^{\infty}\right)\exp\left(-\ell_{i}t\right)+\xi_{i}^{\infty},\hspace{1em}\forall i=1,2,\ldots,4\label{eq:NAV_Presc}
\end{equation}
with $\xi_{i}\left(0\right)=\xi_{i}^{0}>0$ being upper bound of the
known large set, $\xi_{i}^{\infty}>0$ being upper bound of the small
set, and $\ell_{i}>0$ being convergence rate of $\xi_{i}\left(t\right)$
from $\xi_{i}^{0}$ to $\xi_{i}^{\infty}$. Define the error $e_{i}$
as \cite{bechlioulis2008robust}:
\begin{equation}
e_{i}=\xi_{i}\mathcal{N}(E_{i})\label{eq:NAV_e_Trans}
\end{equation}
where $E_{i}\in\mathbb{R}$ denotes the transformed error which is
unconstrained, and $\mathcal{N}(E_{i})$ denotes a smooth function
which is differentiable, strictly increasing, and bounded, for more
information see \cite{hashim2019SO3Wiley,hashim2020TITS_SLAM}. The
inverse transformation of \eqref{eq:NAV_Trans1} is given by
\begin{equation}
E_{i}=\mathcal{N}^{-1}(e_{i}/\xi_{i})\label{eq:NAV_Trans1}
\end{equation}
The inverse transformation in \eqref{eq:NAV_Trans1} is defined by
\cite{hashim2019SO3Wiley,hashim2020TITS_SLAM,bechlioulis2008robust}
\begin{equation}
\begin{aligned}E_{i}= & \frac{1}{2}\begin{cases}
\text{ln}\frac{\underline{\delta}_{i}+e_{i}/\xi_{i}}{\bar{\delta}_{i}-e_{i}/\xi_{i}}, & \bar{\delta}_{i}>\underline{\delta}_{i}\text{ if }e_{i}\left(0\right)\geq0\\
\text{ln}\frac{\underline{\delta}_{i}+e_{i}/\xi_{i}}{\bar{\delta}_{i}-e_{i}/\xi_{i}}, & \underline{\delta}_{i}>\bar{\delta}_{i}\text{ if }e_{i}\left(0\right)<0
\end{cases}\end{aligned}
\label{eq:NAV_trans2}
\end{equation}
Let us define 
\begin{equation}
\begin{split}\Delta_{i} & =\frac{1}{2\xi_{i}}\frac{\partial\mathcal{N}^{-1}\left(e_{i}/\xi_{i}\right)}{\partial\left(e_{i}/\xi_{i}\right)}=\frac{1}{2\xi_{i}}\left(\frac{1}{\underline{\delta}_{i}+e_{i}/\xi_{i}}+\frac{1}{\bar{\delta}_{i}-e_{i}/\xi_{i}}\right)\end{split}
\label{eq:SE3PPF_mu}
\end{equation}
Accordingly, the dynamics of $\dot{E}_{i}$ are
\begin{align}
\dot{E}_{i} & =\Delta_{i}\left(\frac{d}{dt}e_{i}-\frac{\dot{\xi}_{i}}{\xi_{i}}e_{i}\right)\label{eq:SE3PPF_Trans_dot1}
\end{align}
Or to put simply
\begin{equation}
\dot{E}=\left[\begin{array}{cc}
\Delta_{R} & 0_{1\times3}\\
0_{3\times1} & \Delta_{P}
\end{array}\right]\left(\frac{d}{dt}e-\left[\begin{array}{cc}
\mu_{R} & 0_{1\times3}\\
0_{3\times1} & \mu_{P}
\end{array}\right]e\right)\label{eq:SE3PPF_Trans_dot}
\end{equation}
such that $\mu_{R}=\dot{\xi}_{1}/\xi_{1}$, $\mu_{P}={\rm diag}(\dot{\xi}_{2}/\xi_{2},\dot{\xi}_{3}/\xi_{3},\dot{\xi}_{4}/\xi_{4})$,
$\Delta_{R}=\Delta_{1}$, and $\Delta_{P}={\rm diag}(\Delta_{2},\Delta_{3},\Delta_{4})$
where $\mu_{R},\Delta_{R}\in\mathbb{R}$ and $\mu_{P},\Delta_{P}\in\mathbb{R}^{3\times3}$.
For more information of orientation, pose, and SLAM filters with prescribed
performance visit \cite{hashim2019SO3Wiley,hashim2020TITS_SLAM}.

\subsection{Nonlinear Navigation Filter with Bias Compensation\label{subsec:Non-Nav-Observer2}}

In this Subsection the unknown bias inevitably present in measurements
of angular velocity and acceleration is accounted for. From \eqref{eq:NAV_XVelcoity},
let $\hat{b}_{\Omega}$ and $\hat{b}_{a}$ denote the estimates of
$b_{\Omega}$ and $b_{a}$, respectively. Define the bias error as
\begin{equation}
\begin{cases}
\tilde{b}_{\Omega} & =b_{\Omega}-\hat{b}_{\Omega}\\
\tilde{b}_{a} & =\hat{b}_{a}-b_{a}
\end{cases}\label{eq:NAV_bias_err}
\end{equation}
Consider the following nonlinear filter design on the Lie Group of
$\mathbb{SE}_{2}\left(3\right)$:{\small{}
	\begin{equation}
	\begin{cases}
	\dot{\hat{R}} & =\hat{R}[\Omega_{m}-\hat{b}_{\Omega}]_{\times}-\left[w_{\Omega}\right]_{\times}\hat{R}\\
	\dot{\hat{P}} & =\hat{V}-\left[w_{\Omega}\right]_{\times}\hat{P}-w_{V}\\
	\dot{\hat{V}} & =\hat{R}(a_{m}-\hat{b}_{a})-\left[w_{\Omega}\right]_{\times}\hat{V}-w_{a}
	\end{cases},\hspace{0.5em}\underbrace{\dot{\hat{X}}=\hat{X}U_{m}-W\hat{X}}_{\text{Compact form}}\label{eq:NAV_Est_Dot-2}
	\end{equation}
}such that
\begin{equation}
\begin{cases}
w_{\Omega} & =-k_{w}(E_{R}\Delta_{R}+1)\boldsymbol{\Upsilon}(M\tilde{R})\\
w_{V} & =\left[p_{c}-\tilde{R}^{\top}\tilde{P}_{\varepsilon}\right]_{\times}w_{\Omega}-\ell_{P}\tilde{R}^{\top}\tilde{P}_{\varepsilon}-k_{v}\Delta_{P}E_{P}\\
w_{a} & =-\overrightarrow{\mathtt{g}}+k_{a}\left(\delta\left[w_{\Omega}\right]_{\times}-\Delta_{P}\right)E_{P}\\
\dot{\hat{b}}_{\Omega} & =-\gamma_{b}(\Delta_{R}E_{R}+1)\hat{R}^{\top}\boldsymbol{\Upsilon}(M\tilde{R})\\
\dot{\hat{b}}_{a} & =-\gamma_{a}\delta\hat{R}^{\top}E_{P}
\end{cases}\label{eq:NAV_Filter2_Detailed}
\end{equation}
with $k_{w}$, $k_{v}$, $k_{a}$, $\delta$, $\ell_{P}$, $\gamma_{b}$,
and $\gamma_{a}$ being positive constants, $w_{\Omega}$, $w_{V}$,
and $w_{a}$ denoting correction factors $\forall$ $w_{\Omega},w_{V},w_{a}\in\mathbb{R}^{3}$,
$\boldsymbol{\Upsilon}(M\tilde{R})=\mathbf{vex}(\boldsymbol{\mathcal{P}}_{a}(M\tilde{R}))$,
$E_{R}=E_{1}$, and $E_{P}=[E_{2},E_{3},E_{4}]^{\top}$. Also, $U_{m}=u([\Omega_{m}-\hat{b}_{\Omega}\text{\ensuremath{]_{\times}}},0_{3\times1},a_{m}-\hat{b}_{a},1)\in\mathcal{U}_{\mathcal{M}}$
and $W=u([w_{\Omega}\text{\ensuremath{]_{\times}}},w_{V},w_{a},1)]\in\mathcal{U}_{\mathcal{M}}$.
\begin{thm}
	\label{thm:Theorem2}Consider the navigation dynamics in \eqref{eq:NAV_True_dot}
	coupled with feature measurements (output $\overline{y}_{i}=X^{-1}\overline{p}_{i}$)
	for all $i=1,2,\ldots,n$, angular velocity measurements ($\Omega_{m}=\Omega+b_{\Omega}$),
	and acceleration measurements ($a_{m}=a+b_{a}$). Let Assumption \ref{Assum:NAV_1Landmark}
	hold. Combine the nonlinear filter design in \eqref{eq:NAV_Est_Dot-2}
	and \eqref{eq:NAV_Filter2_Detailed} with the measurements of $\overline{y}_{i}$,
	$\Omega_{m}$, and $a_{m}$. Let the design parameters $k_{w}$, $k_{v}$,
	$k_{a}$, $\ell_{P}$, $\gamma_{b}$, $\gamma_{a}$, $\delta$, $\underline{\delta}_{i}=\bar{\delta}_{i}>|e_{i}(0)|$,
	$\xi_{i}^{0}>|e_{i}(0)|$ and $\xi_{i}^{\infty}$ be selected as positive
	constants. Define the set
	\begin{align}
	\mathcal{S}= & \{\left.(E_{R},E_{P},\tilde{V},\tilde{b}_{\Omega},\tilde{b}_{a})\in\mathbb{R}\times\mathbb{R}^{3}\times\mathbb{R}^{3}\times\mathbb{R}^{3}\times\mathbb{R}^{3}\right|\nonumber \\
	& \hspace{5em}E_{R}=0,E_{P}=\tilde{V}=\tilde{b}_{\Omega}=\tilde{b}_{a}=0_{3\times1}\}\label{eq:NAV_Set2}
	\end{align}
	given that $E_{R},E_{P}\in\mathcal{L}_{\infty}$ and $\tilde{R}(0)$
	does not belong to the unattractive set defined in \cite{hashim2019SO3Wiley}.
	Then, 1) the errors $E_{R}$, $E_{P}$, $\tilde{V}$, $\tilde{b}_{\Omega}$,
	and $\tilde{b}_{a}$ converge asymptotically to $\mathcal{S}$ in
	\eqref{eq:NAV_Set2}, 2) the trajectory of $\tilde{R}$ converges
	asymptotically to $\mathbf{I}_{3}$, and 3) the trajectories of $\tilde{P}$
	converge asymptotically to the origin.
\end{thm}
\begin{proof}From \eqref{eq:NAV_True_dot} and \eqref{eq:NAV_Est_Dot-2},
	the orientation error dynamics are
	\begin{align}
	\dot{\tilde{R}} & =-\tilde{R}[\hat{R}(\Omega_{m}-\Omega)]_{\times}+\tilde{R}[w_{\Omega}\text{\ensuremath{]_{\times}}}\label{eq:NAV_R_error_Dot}
	\end{align}
	where $[\hat{R}\Omega]_{\times}=\hat{R}\left[\Omega\right]_{\times}\hat{R}^{\top}$.
	From \eqref{eq:NAV_R_error_Dot}, \eqref{eq:NAV_True_dot}, and \eqref{eq:NAV_Est_Dot-2},
	the position error dynamics are
	\begin{align}
	\dot{\tilde{P}} & =\tilde{V}-\tilde{R}[\hat{P}]_{\times}\hat{R}(\Omega_{m}-\Omega)+\tilde{R}w_{V}\label{eq:NAV_P_error_Dot}
	\end{align}
	From \eqref{eq:NAV_R_error_Dot}, \eqref{eq:NAV_True_dot}, and \eqref{eq:NAV_Est_Dot-2},
	the velocity error dynamics are
	\begin{align}
	\dot{\tilde{V}} & =-R(a_{m}-a)-\tilde{R}[\hat{V}]_{\times}\hat{R}(\Omega_{m}-\Omega)+\overrightarrow{\mathtt{g}}+\tilde{R}w_{a}\label{eq:NAV_V_error_Dot}
	\end{align}
	Given that $\Omega_{m}=\Omega+b_{\Omega}$ and $a_{m}=a+b_{a}$, one
	has
	
	\begin{equation}
	\begin{cases}
	\frac{d}{dt}||M\tilde{R}||_{{\rm I}} & =-\frac{1}{2}\boldsymbol{\Upsilon}(M\tilde{R})^{\top}(\hat{R}\tilde{b}_{\Omega}-w_{\Omega})\\
	\frac{d}{dt}\tilde{R}^{\top}\tilde{P}_{\varepsilon} & =\tilde{R}^{\top}\tilde{V}-\left[\hat{P}-p_{c}+\tilde{R}^{\top}\tilde{P}_{\varepsilon}\right]_{\times}\hat{R}\tilde{b}_{\Omega}\\
	& \hspace{1em}-\left[p_{c}-\tilde{R}^{\top}\tilde{P}_{\varepsilon}\right]_{\times}w_{\Omega}+w_{V}\\
	\frac{d}{dt}\tilde{R}^{\top}\tilde{V} & =-\left[\tilde{R}^{\top}V\right]_{\times}\hat{R}\tilde{b}_{\Omega}-\left[w_{\Omega}\right]_{\times}\tilde{R}^{\top}\tilde{V}\\
	& \hspace{1em}+\hat{R}\tilde{b}_{a}+\tilde{R}^{\top}\overrightarrow{\mathtt{g}}+w_{a}
	\end{cases}\label{eq:NAV_Filter2_Error_dot}
	\end{equation}
	From \eqref{eq:NAV_Filter2_Error_dot} and \eqref{eq:SE3PPF_Trans_dot},
	one obtains
	\begin{equation}
	\begin{cases}
	\dot{E}_{R} & =\Delta_{R}\left(\frac{d}{dt}||M\tilde{R}||_{{\rm I}}-\mu_{R}||M\tilde{R}||_{{\rm I}}\right)\\
	\dot{E}_{P} & =\Delta_{P}\left(\frac{d}{dt}\tilde{R}^{\top}\tilde{P}_{\varepsilon}-\mu_{P}\tilde{R}^{\top}\tilde{P}_{\varepsilon}\right)
	\end{cases}\label{eq:NAV_Filter2_Trans_dot}
	\end{equation}
	Consider selecting the following Lyapunov function candidate $\mathcal{L}_{T}=\mathcal{L}_{T}(||M\tilde{R}||_{{\rm I}},E_{R},E_{P},\tilde{R}^{\top}\tilde{V},\hat{R}\tilde{b}_{\Omega},\hat{R}\tilde{b}_{a})$
	\begin{equation}
	\mathcal{L}_{T}=\mathcal{L}_{R}+\mathcal{L}_{PV}\label{eq:NAV_LyapT_2}
	\end{equation}
	Let $\overline{\lambda}_{\overline{\mathbf{M}}}=\overline{\lambda}(\overline{\mathbf{M}})$.
	The function $\mathcal{L}_{R}$ is selected as
	\begin{align}
	\mathcal{L}_{R} & =E_{R}^{2}+\mathcal{L}_{1}\label{eq:NAV_LyapR_2}
	\end{align}
	where $\mathcal{L}_{Rb_{\Omega}}=2||M\tilde{R}||_{{\rm I}}+\frac{1}{2\gamma_{b}}||\tilde{b}_{\Omega}||^{2}+\frac{1}{2\overline{\gamma}_{b}\overline{\lambda}_{\overline{\mathbf{M}}}}\boldsymbol{\Upsilon}(M\tilde{R})^{\top}\hat{R}\tilde{b}_{\Omega}$.
	$\mathcal{L}_{Rb_{\Omega}}$ in \eqref{eq:NAV_LyapR_2} follows
	\[
	\varepsilon_{1}^{\top}\underbrace{\left[\begin{array}{cc}
		2 & \frac{-1}{4\overline{\gamma}_{b}\overline{\lambda}_{\overline{\mathbf{M}}}}\\
		\frac{-1}{4\overline{\gamma}_{b}\overline{\lambda}_{\overline{\mathbf{M}}}} & \frac{1}{2\gamma_{b}}
		\end{array}\right]}_{\underline{A}_{1}}\varepsilon_{1}\leq\mathcal{L}_{1}\leq\varepsilon_{1}^{\top}\underbrace{\left[\begin{array}{cc}
		2 & \frac{1}{4\overline{\gamma}_{b}\overline{\lambda}_{\overline{\mathbf{M}}}}\\
		\frac{1}{4\overline{\gamma}_{b}\overline{\lambda}_{\overline{\mathbf{M}}}} & \frac{1}{2\gamma_{b}}
		\end{array}\right]}_{\overline{A}_{1}}\varepsilon_{1}
	\]
	with $\varepsilon_{1}=\left[\sqrt{||M\tilde{R}||_{{\rm I}}},||\hat{R}\tilde{b}_{\Omega}||\right]^{\top}$
	and $\sqrt{2\overline{\lambda}_{\overline{\mathbf{M}}}||\tilde{R}M||_{{\rm I}}}\geq||\boldsymbol{\Upsilon}(M\tilde{R})||$
	as defined in \eqref{eq:SLAM_lemm1}, Lemma \ref{Lemm:SLAM_Lemma1}.
	Both $\underline{A}_{1}$ and $\overline{A}_{1}$ become positive
	given that
	\begin{equation}
	\frac{1}{\gamma_{b}}>\frac{1}{16\overline{\gamma}_{b}^{2}\overline{\lambda}_{\overline{\mathbf{M}}}^{2}}\label{eq:NAV_Filter2_Select1}
	\end{equation}
	Let $\gamma_{b}$ be selected as in \eqref{eq:NAV_Filter2_Select1}.
	Recall \eqref{eq:NAV_Filter2_Trans_dot} and \eqref{eq:NAV_Filter2_Error_dot}.
	The derivative of \eqref{eq:NAV_LyapR_2} is
	\begin{align*}
	\dot{\mathcal{L}}_{R} & \leq-c_{R}(\Delta_{R}^{2}E_{R}^{2}+1)||M\tilde{R}||_{{\rm I}}-c_{b}||\hat{R}\tilde{b}_{\Omega}||^{2}\\
	& +\frac{c_{\Omega}}{\overline{\gamma}_{b}\sqrt{2\overline{\lambda}_{\overline{\mathbf{M}}}}}\sqrt{||M\tilde{R}||_{{\rm I}}}\,||\hat{R}\tilde{b}_{\Omega}||
	\end{align*}
	which can be expressed as
	\begin{align}
	\dot{\mathcal{L}}_{R}\leq & -\varepsilon_{1}^{\top}\underbrace{\left[\begin{array}{cc}
		c_{R} & \frac{c_{\Omega}}{\overline{\gamma}_{b}\sqrt{2\overline{\lambda}_{\overline{\mathbf{M}}}}}\\
		\frac{c_{\Omega}}{\overline{\gamma}_{b}\sqrt{2\overline{\lambda}_{\overline{\mathbf{M}}}}} & c_{b}
		\end{array}\right]}_{A_{2}}\varepsilon_{1}\nonumber \\
	& -c_{R}E_{R}^{2}\Delta_{R}^{2}||M\tilde{R}||_{{\rm I}}\label{eq:NAV_LyapR_2dot_3}
	\end{align}
	where $\varepsilon_{1}=\left[\sqrt{||M\tilde{R}||_{{\rm I}}},||\hat{R}\tilde{b}_{\Omega}||\right]^{\top}$.
	It can be deduced that $A_{2}$ become positive by selecting $c_{R}c_{b}\geq\frac{c_{\Omega}^{2}}{2\overline{\gamma}_{b}^{2}\overline{\lambda}_{\overline{\mathbf{M}}}}$
	such that
	\begin{align}
	\overline{\gamma}_{b} & \geq\sqrt{\frac{8c_{\Omega}^{2}(1+{\rm Tr}\{\tilde{R}\})+K_{M}}{2\sqrt{3}(k_{w}-1)\underline{\lambda}_{\overline{\mathbf{M}}}\overline{\lambda}_{\overline{\mathbf{M}}}}}\label{eq:NAV_2Select1}
	\end{align}
	where $K_{M}=(\sqrt{3}\gamma_{b}+(2(k_{w}-1)\overline{\lambda}_{\overline{\mathbf{M}}}-\gamma_{b})k_{w}^{2})\underline{\lambda}_{\overline{\mathbf{M}}}$.
	Consider selecting $\overline{\gamma}_{b}$ as in \eqref{eq:NAV_2Select1}
	and defining $\underline{\lambda}_{A_{2}}=\underline{\lambda}(A_{2})$.
	Thus, $\dot{\mathcal{L}}_{R}$ in \eqref{eq:NAV_LyapR_2dot_3} follows
	the inequality below
	\begin{align}
	\dot{\mathcal{L}}_{R}\leq-\underline{\lambda}_{A_{2}}||\varepsilon_{R}||^{2} & -c_{R}E_{R}^{2}\Delta_{R}^{2}||M\tilde{R}||_{{\rm I}}\label{eq:NAV_LyapR_2dot_Final}
	\end{align}
	Let us turn our attention to the second portion of the $\mathcal{L}_{T}$
	definition in \eqref{eq:NAV_LyapT_2}. Define the following Lyapunov
	function candidate:
	\begin{align}
	\mathcal{L}_{PV} & =\frac{1}{2}||E_{P}||^{2}+\frac{1}{2k_{a}}||\tilde{V}^{\top}\tilde{R}||^{2}+\frac{1}{2\gamma_{a}}\left\Vert \hat{R}\tilde{b}_{a}\right\Vert ^{2}\nonumber \\
	& -\delta E_{P}^{\top}\tilde{R}^{\top}\tilde{V}+\delta_{a}\tilde{b}_{a}^{\top}\hat{R}^{\top}\tilde{R}^{\top}\tilde{V}\label{eq:NAV_LyapPV_2}
	\end{align}
	It is straightforward to show that $\mathcal{L}_{PV}$ in \eqref{eq:NAV_LyapPV_2}
	follows{\small{}
		\[
		\varepsilon_{2}^{\top}\underbrace{\left[\begin{array}{ccc}
			\frac{1}{2} & \frac{-\delta}{2} & 0\\
			\frac{-\delta}{2} & \frac{1}{2k_{a}} & \frac{-\delta_{a}}{2}\\
			0 & \frac{-\delta_{a}}{2} & \frac{1}{2\gamma_{a}}
			\end{array}\right]}_{\underline{A}_{3}}\varepsilon_{2}\leq\mathcal{L}_{PV}\leq\varepsilon_{2}^{\top}\underbrace{\left[\begin{array}{ccc}
			\frac{1}{2} & \frac{\delta}{2} & 0\\
			\frac{\delta}{2} & \frac{1}{2k_{a}} & \frac{\delta_{a}}{2}\\
			0 & \frac{\delta_{a}}{2} & \frac{1}{2\gamma_{a}}
			\end{array}\right]}_{\overline{A}_{3}}\varepsilon_{2}
		\]
	}with $\varepsilon_{2}=\left[||E_{P}||,||\tilde{V}^{\top}\tilde{R}||,||\hat{R}\tilde{b}_{a}||\right]^{\top}$.
	It becomes apparent that $\text{\ensuremath{\underline{A}}}_{3}$,
	$\overline{A}_{3}$, and in turn $\mathcal{L}_{PV}$ can be made positive,
	if the following holds $\frac{1}{8\gamma_{a}}(\frac{1}{k_{a}}-\delta^{2})-\frac{\delta_{a}^{2}}{8}>0$
	which can be enabled by selecting $\delta,\delta_{a},k_{a},\gamma_{a}>0$
	and 
	\begin{align}
	k_{a} & <\frac{1}{\delta_{a}^{2}\gamma_{a}+\delta^{2}}\label{eq:NAV_2Select2}
	\end{align}
	Let $k_{a}$ be selected as in \eqref{eq:NAV_2Select2}. Recall \eqref{eq:NAV_Filter2_Trans_dot}
	and \eqref{eq:NAV_Filter2_Error_dot}. The time derivative of $\mathcal{L}_{PV}$
	in \eqref{eq:NAV_LyapPV_2} is as follows:
	\begin{align}
	\dot{\mathcal{L}}_{PV}\leq & -\varepsilon_{PV}^{\top}\underbrace{\left[\begin{array}{cc}
		k_{v}c_{1}-\delta k_{a}c_{2} & \frac{\delta k_{v}c_{4}}{2}\\
		\frac{\delta k_{v}c_{4}}{2} & \delta c_{3}
		\end{array}\right]}_{A_{4}}\varepsilon_{PV}\nonumber \\
	& -\delta_{a}||\hat{R}\tilde{b}_{a}||^{2}+(\delta_{a}c_{a}||\hat{R}\tilde{b}_{a}||+c_{b}||\hat{R}\tilde{b}_{\Omega}||)||\tilde{R}^{\top}\tilde{V}||\nonumber \\
	& +c_{g}(||\tilde{V}||+||E_{P}||+||\hat{R}\tilde{b}_{a}||)||\mathbf{I}_{3}-\tilde{R}||_{F}\nonumber \\
	& +c_{b}||\hat{R}\tilde{b}_{\Omega}||\,||E_{P}||+\delta_{a}c_{a}||\hat{R}\tilde{b}_{a}||\,||E_{P}||\label{eq:NAV_LyapPV_2dot_5}
	\end{align}
	where $\varepsilon_{PV}=\left[||E_{P}||,||\tilde{V}^{\top}\tilde{R}||\right]^{\top}$.
	It becomes clear that $A_{4}$ in \eqref{eq:NAV_LyapPV_2dot_5} can
	be made positive by selecting
	\begin{equation}
	\frac{1}{\delta}>\frac{k_{v}^{2}c_{4}^{2}+4k_{a}c_{3}c_{2}}{4c_{3}c_{1}k_{v}}\label{eq:NAV_2select_4}
	\end{equation}
	Let $\delta$ be selected as in \eqref{eq:NAV_2select_4} and define
	$\underline{\lambda}_{A_{4}}=\underline{\lambda}(A_{4})$. One may
	rewrite the inequality in \eqref{eq:NAV_LyapPV_2dot_5} as follows:
	\begin{align}
	\dot{\mathcal{L}}_{PV}\leq & -\left[\begin{array}{c}
	||\varepsilon_{PV}||\\
	||\hat{R}\tilde{b}_{a}||
	\end{array}\right]^{\top}\underbrace{\left[\begin{array}{cc}
		\underline{\lambda}_{A_{4}} & \frac{\delta_{a}c_{a}}{2}\\
		\frac{\delta_{a}c_{a}}{2} & \delta_{a}
		\end{array}\right]}_{A_{5}}\left[\begin{array}{c}
	||\varepsilon_{PV}||\\
	||\hat{R}\tilde{b}_{a}||
	\end{array}\right]\nonumber \\
	& +c_{g}(||\tilde{V}||+||E_{P}||+||\hat{R}\tilde{b}_{a}||)||\mathbf{I}_{3}-\tilde{R}||_{F}\nonumber \\
	& +c_{b}||\hat{R}\tilde{b}_{\Omega}||\,||E_{P}||+c_{b}||\hat{R}\tilde{b}_{\Omega}||\,||\tilde{R}^{\top}\tilde{V}||\label{eq:NAV_LyapPV_2dot_6}
	\end{align}
	It can be shown that $A_{5}$ is positive if the following holds:
	\begin{equation}
	\delta_{a}<\frac{4\underline{\lambda}_{A_{4}}}{c_{a}^{2}}\label{eq:NAV_2select_5}
	\end{equation}
	Consider selecting $\delta_{a}$ as in \eqref{eq:NAV_2select_5} and
	defining $\underline{\lambda}_{A_{5}}=\underline{\lambda}(A_{5})$.
	Define $\varepsilon_{T}=\left[||\varepsilon_{Pb}||,\sqrt{||M\tilde{R}||_{{\rm I}}},||\hat{R}\tilde{b}_{\Omega}||\right]^{\top}$
	and let $c_{n}=\max\{4\overline{\lambda}_{M}c_{g},c_{b}\}$. Hence,
	one obtains
	\begin{align}
	\dot{\mathcal{L}}_{T}\leq & -\varepsilon_{T}^{\top}\underbrace{\left[\begin{array}{cc}
		\underline{\lambda}_{A_{5}} & \frac{c_{n}}{2}\mathbf{1}_{1\times2}\\
		\frac{c_{n}}{2}\mathbf{1}_{2\times1} & \underline{\lambda}_{A_{2}}\mathbf{I}_{2}
		\end{array}\right]}_{A_{T}}\varepsilon_{T}\nonumber \\
	& -\frac{\underline{\lambda}_{\overline{\mathbf{M}}}k_{w}E_{R}^{2}\Delta_{R}^{2}}{4(1+{\rm Tr}\{\tilde{R}\})}||M\tilde{R}||_{{\rm I}}\label{eq:NAV_LyapPV_2dot_8}
	\end{align}
	Select the parameters of $A_{2}$ and $A_{5}$ such that $4\underline{\lambda}_{A_{5}}\underline{\lambda}_{A_{2}}>c_{n}^{2}$
	guaranteeing the positive definiteness of $A_{T}$. Accordingly, the
	following result is obtained:
	\begin{equation}
	\dot{\mathcal{L}}_{T}\leq-\underline{\lambda}(A_{T})\,||\varepsilon_{T}||^{2}-\frac{\underline{\lambda}_{\overline{\mathbf{M}}}k_{w}E_{R}^{2}\Delta_{R}^{2}}{4(1+{\rm Tr}\{\tilde{R}\})}||M\tilde{R}||_{{\rm I}}\label{eq:NAV_LyapPV_2dot_Final}
	\end{equation}
	Recall that $\varepsilon_{T}=\left[||\varepsilon_{Pb}||,\sqrt{||M\tilde{R}||_{{\rm I}}},||\hat{R}\tilde{b}_{\Omega}||\right]$
	where $\varepsilon_{Pb}=\left[||\varepsilon_{PV}||,||\hat{R}\tilde{b}_{a}||\right]^{\top}$,
	and $\varepsilon_{PV}=\left[||E_{P}||,||\tilde{V}^{\top}\tilde{R}||\right]^{\top}$.
	In accordance with the definition of $\mathcal{L}_{T}$ in \eqref{eq:NAV_LyapT_2},
	the inequality in \eqref{eq:NAV_LyapPV_2dot_Final} indicates that
	$\dot{\mathcal{L}}_{T}$ is negative for all $\mathcal{L}_{T}>0$,
	and $\dot{\mathcal{L}}_{T}=0$ only at $\mathcal{L}_{T}=0$. This
	ensures asymptotic convergence of $\mathcal{L}_{T}$ enabling the
	guaranteed performance of transient and steady-state error in \eqref{eq:NAV_Vec_error}.
	This completes the proof.\end{proof}

For implementation purposes, the nonlinear navigation filter with
guaranteed performance for unknown bias outlined in \eqref{eq:NAV_Est_Dot-2}
and detailed in \eqref{eq:NAV_Filter2_Detailed} is presented in discrete
form. Let $\Delta t$ denote a small sample time step. The complete
discrete implementation steps are detailed in Algorithm \ref{alg:Alg_Disc}. 

\begin{algorithm}
	\caption{\label{alg:Alg_Disc}Discrete nonlinear filter.}
	
	\textbf{Initialization}:
	\begin{enumerate}
		\item[{\footnotesize{}1:}] Set $\hat{R}[0]\in\mathbb{SO}\left(3\right)$, $\hat{P}[0]\in\mathbb{R}^{3}$,
		$\hat{V}[0]\in\mathbb{R}^{3}$, and $\hat{b}_{\Omega}[0]=\hat{b}_{a}[0]=0_{3\times1}$
		\item[{\footnotesize{}2:}] Select $k_{w}$, $k_{v}$, $k_{a}$, $\ell_{P}$, $\gamma_{b}$,
		$\gamma_{a}$, $\delta$, $\underline{\delta}_{i}=\bar{\delta}_{i}>|e_{i}(0)|$,
		$\xi_{i}^{0}>|e_{i}(0)|$ and $\xi_{i}^{\infty}$ as positive constants,
		and the sample $k=0$
	\end{enumerate}
	\textbf{while }(1)\textbf{ do}
	\begin{enumerate}
		\item[] \textcolor{blue}{/{*} Prediction step {*}/}
		\item[{\footnotesize{}3:}] $\hat{X}_{k|k}=\left[\begin{array}{ccc}
		\hat{R}_{k|k} & \hat{P}_{k|k} & \hat{V}_{k|k}\\
		0_{1\times3} & 1 & 0\\
		0_{1\times3} & 0 & 1
		\end{array}\right]$ and $\hat{U}_{k}=\left[\begin{array}{ccc}
		\left[\Omega_{m}[k]-\hat{b}_{\Omega}[k]\right]_{\times} & 0_{3\times1} & a_{m}[k]-\hat{b}_{a}[k]\\
		0_{1\times3} & 0 & 0\\
		0_{1\times3} & 1 & 0
		\end{array}\right]$
		\item[{\footnotesize{}4:}] $\hat{X}_{k+1|k}=\hat{X}_{k|k}\exp(\hat{U}_{k}\Delta t)$
		\item[] \textcolor{blue}{/{*} Update step {*}/}
		\item[{\footnotesize{}5:}] $\begin{cases}
		p_{c} & =\frac{1}{s_{T}}\sum_{i=1}^{n}s_{i}p_{i}[k],\hspace{1em}s_{T}=\sum_{i=1}^{n}s_{i}\\
		M_{k} & =\sum_{i=1}^{n}s_{i}p_{i}[k]p_{i}^{\top}[k]-s_{T}p_{c}p_{c}^{\top}\\
		M\tilde{R}_{k} & =\sum_{i=1}^{n}s_{i}\left(p_{i}[k]-p_{c}\right)y_{i}^{\top}[k]\hat{R}_{k+1|k}^{\top}\\
		\tilde{R}^{\top}\tilde{P}_{\varepsilon}[k] & =\frac{\sum_{i=1}^{n}s_{i}(p_{i}[k]-\hat{R}_{k+1|k}y_{i}[k]-\hat{P}_{k+1|k})}{s_{T}}
		\end{cases}$
		\item[{\footnotesize{}6:}] $[e_{1}[k],e_{2}[k],e_{3}[k],e_{4}[k]]^{\top}=\left[||M\tilde{R}_{k}||_{{\rm I}},\tilde{P}_{\varepsilon}^{\top}\tilde{R}[k]\right]^{\top}\in\mathbb{R}^{4}$
		\item[{\footnotesize{}7:}] \textbf{for} $i=1:4$ \hspace{0.5cm}\textcolor{blue}{/{*} Guaranteed
			Performance{*}/}
		\item[{\footnotesize{}8:}] \hspace{0.5cm}$\xi_{i}[k]=(\xi_{0}-\xi_{\infty})\exp(-\ell k\Delta t)+\xi_{\infty}$
		\item[{\footnotesize{}9:}] \hspace{0.5cm}\textbf{if} $e_{i}[k]>\xi_{i}[k]$ \textbf{then}
		\item[{\footnotesize{}10:}] \hspace{0.5cm}\hspace{0.5cm}$\xi_{i}[k]=e_{i}[k]+\epsilon$,\hspace{0.5cm}
		\textcolor{blue}{/{*} $\epsilon$ is a small constant {*}/}
		\item[{\footnotesize{}11:}] \hspace{0.5cm}\textbf{end if}
		\item[{\footnotesize{}12:}] \hspace{0.5cm}$\begin{cases}
		E_{i} & =\frac{1}{2}\text{ln}\frac{\underline{\delta}_{i}+e_{i}[k]/\xi_{i}[k]}{\bar{\delta}_{i}-e_{i}[k]/\xi_{i}[k]}\\
		\Delta_{i} & =\frac{1}{2\xi_{i}[k]}\left(\frac{1}{\underline{\delta}_{i}+e_{i}[k]/\xi_{i}[k]}+\frac{1}{\bar{\delta}_{i}-e_{i}[k]/\xi_{i}[k]}\right)
		\end{cases}$
		\item[{\footnotesize{}13:}] \textbf{end for}
		\item[{\footnotesize{}14:}] Set $E_{R}[k]=E_{1}$, $E_{P}[k]=[E_{2},E_{3},E_{4}]^{\top}$, $\Delta_{R}[k]=\Delta_{1}$,
		and $\Delta_{P}[k]={\rm diag}(\Delta_{2},\Delta_{3},\Delta_{4})$.
		\item[{\footnotesize{}15:}] $\begin{cases}
		w_{\Omega}[k] & =-k_{w}(E_{R}[k]\Delta_{R}[k]+1)\boldsymbol{\Upsilon}(M\tilde{R}_{k})\\
		w_{V}[k] & =\left[p_{c}[k]-\tilde{R}^{\top}\tilde{P}_{\varepsilon}[k]\right]_{\times}w_{\Omega}[k]\\
		& \hspace{1em}-\ell_{P}\tilde{R}^{\top}\tilde{P}_{\varepsilon}[k]-k_{v}\Delta_{P}[k]E_{P}[k]\\
		w_{a}[k] & =-\overrightarrow{\mathtt{g}}+k_{a}\left(\delta\left[w_{\Omega}[k]\right]_{\times}-\Delta_{P}[k]\right)E_{P}[k]\\
		\hat{b}_{\Omega}[k+1] & =\hat{b}_{\Omega}[k]-\Delta t\gamma_{b}\hat{R}_{k+1|k}^{\top}\boldsymbol{\Upsilon}(M\tilde{R}_{k})\\
		& \hspace{1em}-\Delta t\gamma_{b}\Delta_{R}[k]E_{R}[k]\hat{R}_{k+1|k}^{\top}\boldsymbol{\Upsilon}(M\tilde{R}_{k})\\
		\hat{b}_{a}[k+1] & =\hat{b}_{a}[k]-\Delta t\gamma_{a}\delta\hat{R}_{k+1|k}^{\top}E_{P}[k]
		\end{cases}$
		\item[{\footnotesize{}16:}] $W_{k}=\left[\begin{array}{ccc}
		\left[w_{\Omega}[k]\right]_{\times} & w_{V}[k] & w_{a}[k]\\
		0_{1\times3} & 0 & 0\\
		0_{1\times3} & 1 & 0
		\end{array}\right]$
		\item[{\footnotesize{}17:}] $\hat{X}_{k+1|k+1}=\exp(-W_{k}\Delta t)\hat{X}_{k+1|k}$
		\item[{\footnotesize{}18:}] $k=k+1$
	\end{enumerate}
	\textbf{end while}
\end{algorithm}

\section{Experimental Results \label{sec:SE3_Simulations}}

In this Section, we test the proposed navigation filter against real-world
data obtained from the EuRoC dataset \cite{Burri2016Euroc}. The dataset
in \cite{Burri2016Euroc} is composed of ground truth that describes
the quadrotor true trajectory, uncertain measurements collected by
IMU, and stereo images. It should be noted that the stereo images
were recorded by MT9V034 at a sampling rate of 20Hz and the IMU measurements
were obtained at a sampling rate of 200 Hz using ADIS16448. The dataset
does not include values of real world features, as a result, a set
of virtual landmarks were assigned randomly in accordance with Assumption
\ref{Assum:NAV_1Landmark}. The camera parameters were calibrated
through a Stereo Camera Calibrator Application and the images are
not distorted, see Fig. \ref{fig:NAV_Camera}. More details can be
found about EuRoC dataset in \cite{Burri2016Euroc}. Consider the
initial value of orientation, position and linear velocity to be:
\begin{align*}
\hat{R}_{0}=\hat{R}\left(0\right) & =\mathbf{I}_{3},\hspace{1em}\hat{P}\left(0\right)=\hat{V}\left(0\right)=[0,0,0]^{\top}
\end{align*}
Let the initial bias estimate $\hat{b}_{\Omega}\left(0\right)=\hat{b}_{a}\left(0\right)=[0,0,0]^{\top}$
and the design parameters be defined as $k_{w}=3$, $k_{v}=4$, $k_{a}=\ell_{P}=4$,
$\gamma_{b}=2$, $\gamma_{a}=3$, $\delta=0.15$, $\ell=1.2[1,1,1,1]^{\top}$,
$\xi^{\infty}=[0.03,0.08,0.08,0.08]^{\top}$, and 
\[
\xi_{i}^{0}=\underline{\delta}=\bar{\delta}=\left[\begin{array}{c}
1.3||M\tilde{R}(0)||_{{\rm I}}\\
2\tilde{R}(0)^{\top}\tilde{P}_{\varepsilon}(0)
\end{array}\right]+2[0.25,1,1,1]^{\top}
\]
\begin{figure}[H]
	\centering{}\includegraphics[scale=0.2]{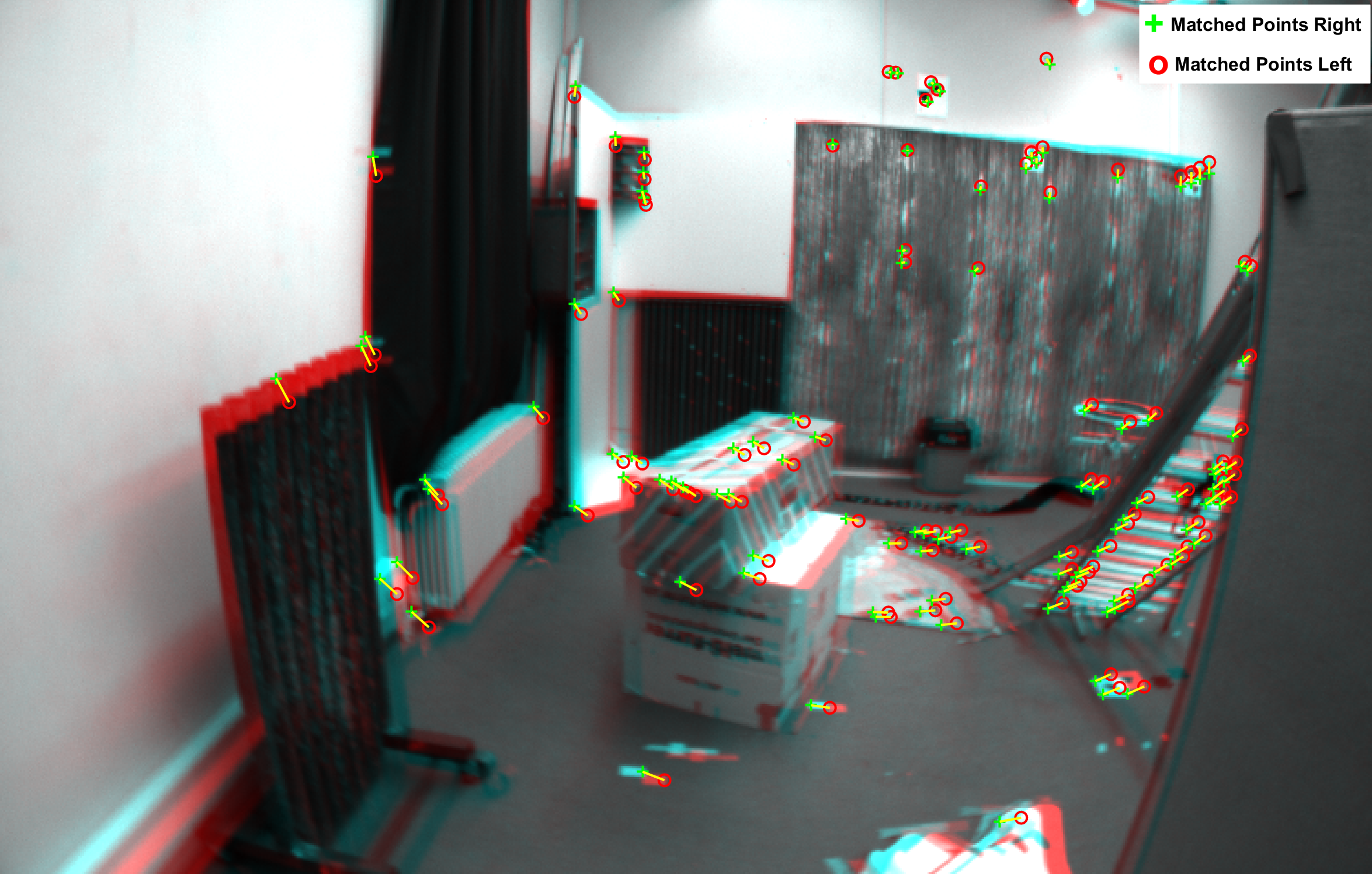}\caption{This image presents an example of right and left feature detection
		and tracking performed by a stereo camera via the Computer Vision
		System Toolbox with MATLAB R2020b. The annotated photograph is adopted
		from the EuRoC dataset \cite{Burri2016Euroc}.}
	\label{fig:NAV_Camera}
\end{figure}

Fig. \ref{fig:NAV_EXPERIMENTAL} shows the experimental performance
of the proposed filter in discrete form. The dataset that was used
in this Section is the Vicon Room 2 01 dataset \cite{Burri2016Euroc}.
Fig. \ref{fig:NAV_EXPERIMENTAL} presents robust convergence starting
from large error in initialization to the neighborhood of the origin.
This can be confirmed from the left portion of Fig. \ref{fig:NAV_EXPERIMENTAL}
which shows accurate tracking performance. Hence, the proposed solution
is promising to be implemented on a cheap electronic kit.

\begin{figure*}
	\centering{}\includegraphics[scale=0.35]{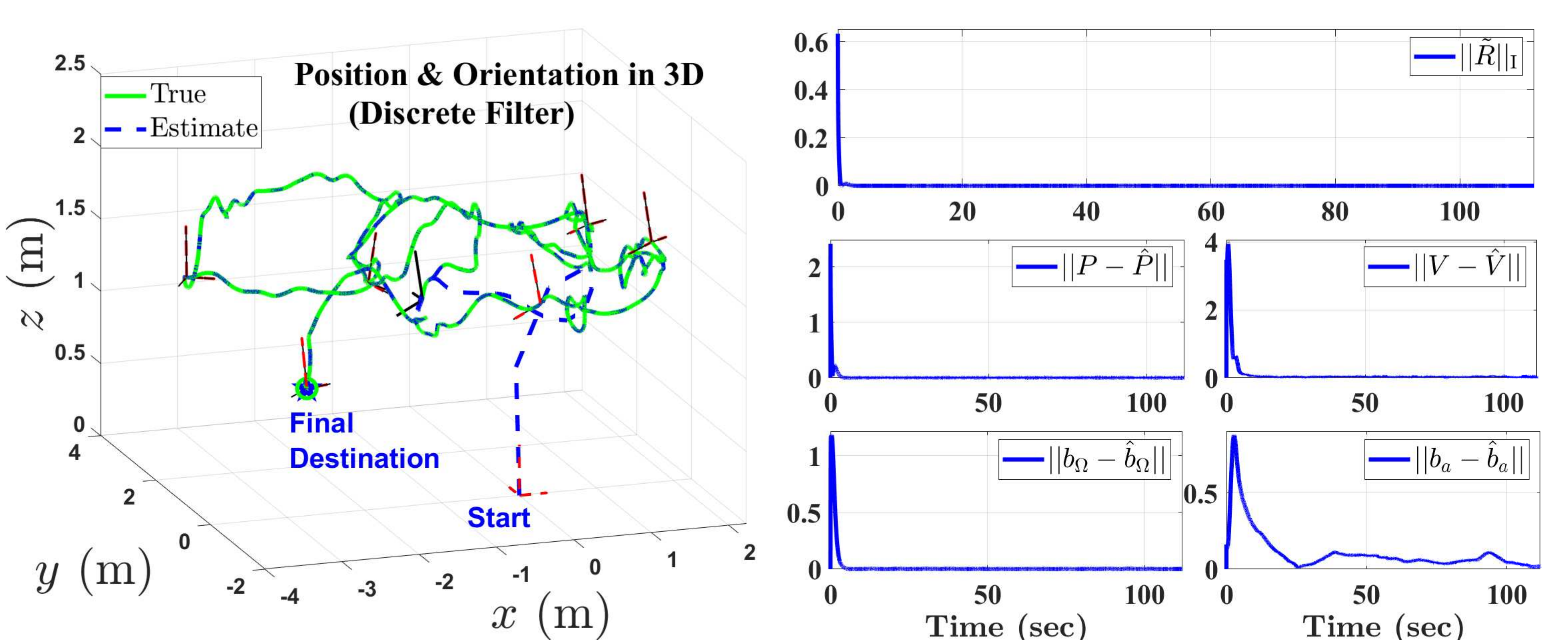}\caption{Dataset Vicon Room 2 01: True trajectory versus estimated trajectory
		plotted using the proposed nonlinear discrete navigation filter. In
		the right portion, the true trajectory is plotted in green solid-line
		while the estimated trajectory is plotted in blue dash-line. The true
		3-axes orientation is plotted in green solid-line while the estimated
		orientation is plotted in red dash-line. The final destination is
		marked with a star. The features are marked by black circles. In the
		right portion, a blue solid-line is used to demonstrate error components
		of: orientation $||R\hat{R}^{\top}||_{{\rm I}}$, position $||P-\hat{P}||$,
		velocity $||V-\hat{V}||$, bias in angular velocity $||b_{\Omega}-\hat{b}_{\Omega}||$,
		and bias in acceleration $||b_{a}-\hat{b}_{a}||$.}
	\label{fig:NAV_EXPERIMENTAL}
\end{figure*}

\section{Conclusion \label{sec:SE3_Conclusion}}

In this paper, the navigation problem of a vehicle traveling in 3D
space has been addressed in deterministic sense. The vehicle's orientation,
position, and linear velocity have been estimated successfully through
a proposed nonlinear filter on $\mathbb{SE}_{2}(3)$. The proposed
solution successfully accounts and compensates for the unknown bias
inevitably present in angular velocity and acceleration measurements.
According to the experimental results, the proposed filter showed
strong tracking capabilities of the unknown pose and estimation of
the unknown linear velocity of the vehicle.

\section*{Acknowledgment}

\subsection*{Appendix\label{subsec:Appendix-A}}
\begin{center}
	\textbf{Quaternion Representation of the Proposed Filter}
	\par\end{center}

\noindent The unit-quaternion is a four-element vector $Q=[q_{0},q_{1},q_{2},q_{3}]^{\top}=[q_{0},q^{\top}]^{\top}\in\mathbb{S}^{3}$
where $q_{0}\in\mathbb{R}$ and $q=[q_{1},q_{2},q_{3}]^{\top}\in\mathbb{R}^{3}$
such that
\[
\mathbb{S}^{3}=\{\left.Q\in\mathbb{R}^{4}\right|||Q||=\sqrt{q_{0}^{2}+q_{1}^{2}+q_{2}^{2}+q_{3}^{2}}=1\}
\]
The inverse of the rotation is represented by the inverse of a unit-quaternion
defined by $Q^{-1}=[\begin{array}{cc}
q_{0} & -q^{\top}\end{array}]^{\top}\in\mathbb{S}^{3}$. Let $\odot$ stand for the quaternion product between two unit-quaternions.
The quaternion multiplication between $Q_{1}=[\begin{array}{cc}
q_{01} & q_{1}^{\top}\end{array}]^{\top}\in\mathbb{S}^{3}$ and $Q_{2}=[\begin{array}{cc}
q_{02} & q_{2}^{\top}\end{array}]^{\top}\in\mathbb{S}^{3}$ is defined by
\begin{align*}
Q_{1}\odot Q_{2} & =\left[\begin{array}{c}
q_{01}q_{02}-q_{1}^{\top}q_{2}\\
q_{01}q_{2}+q_{02}q_{1}+\left[q_{1}\right]_{\times}q_{2}
\end{array}\right]
\end{align*}
For $\Omega\in\mathbb{R}^{3}$, define $\psi(\Omega)=\left[0,\Omega^{\top}\right]^{\top}\in\mathbb{R}^{4}$
and $\mathcal{T}(\psi(\Omega))=\Omega\in\mathbb{R}^{3}$. For more
details of Quaternion to/from $\mathbb{SO}(3)$ representation, visit
\cite{hashim2019AtiitudeSurvey}. Recall the measurements in \eqref{eq:NAV_Set_Measurements}
and define
\[
\begin{cases}
\tilde{y}_{i} & =p_{i}-\mathcal{T}\left(\hat{Q}\odot\psi(y_{i})\odot\hat{Q}^{-1}\right)-\hat{P}\\
\Phi_{q} & =M\tilde{R}=\sum_{i=1}^{n}s_{i}\left(p_{i}-p_{c}\right)y_{i}^{\top}\mathcal{R}_{\hat{Q}}^{\top}\\
{\rm v}_{q} & =\tilde{R}^{\top}\tilde{P}_{\varepsilon}=\frac{1}{s_{T}}\sum_{i=1}^{n}s_{i}\tilde{y}_{i}\\
||M\tilde{R}||_{{\rm I}} & =\frac{1}{4}{\rm Tr}\{M-\Phi_{q}\}
\end{cases}
\]
where $\mathcal{R}_{\hat{Q}}=(\hat{q}_{0}^{2}-||\hat{q}||^{2})\mathbf{I}_{3}+2\hat{q}\hat{q}^{\top}+2\hat{q}_{0}\left[\hat{q}\right]_{\times}$.
Thus, the error vector in \eqref{eq:NAV_Vec_error} is $e=\left[||M\tilde{R}||_{{\rm I}},{\rm v}_{q}^{\top}\right]^{\top}\in\mathbb{R}^{4}$.
The equivalent quaternion representation of the filter can be obtained
as follows:
\begin{equation}
\begin{cases}
\Gamma\left(\Omega\right) & =\left[\begin{array}{cc}
0 & -\Omega^{\top}\\
\Omega & -\left[\Omega\right]_{\times}
\end{array}\right],\hspace{1em}\Psi\left(\Omega\right)=\left[\begin{array}{cc}
0 & -\Omega^{\top}\\
\Omega & \left[\Omega\right]_{\times}
\end{array}\right]\\
\dot{\hat{Q}} & =\frac{1}{2}\Gamma\left(\Omega_{m}-\hat{b}_{\Omega}\right)\hat{Q}-\frac{1}{2}\Psi\left(w_{\Omega}\right)\hat{Q}\\
\dot{\hat{P}} & =\hat{V}-\left[w_{\Omega}\right]_{\times}\hat{P}-w_{V}\\
\dot{\hat{V}} & =\mathcal{T}\left(\hat{Q}\odot\psi(a_{m}-\hat{b}_{a})\odot\hat{Q}^{-1}\right)-\left[w_{\Omega}\right]_{\times}\hat{V}-w_{a}\\
w_{\Omega} & =-k_{w}(E_{R}\Delta_{R}+1)\boldsymbol{\Upsilon}(\Phi_{q})\\
w_{V} & =\left[p_{c}-{\rm v}_{q}\right]_{\times}w_{\Omega}-\ell_{P}{\rm v}_{q}-k_{v}\Delta_{P}E_{P}\\
w_{a} & =-\overrightarrow{\mathtt{g}}+k_{a}\left(\delta\left[w_{\Omega}\right]_{\times}-\Delta_{P}\right)E_{P}\\
\dot{\hat{b}}_{\Omega} & =-\gamma_{b}(\Delta_{R}E_{R}+1)\mathcal{T}\left(\hat{Q}^{-1}\odot\psi(\boldsymbol{\Upsilon}(\Phi_{q}))\odot\hat{Q}\right)\\
\dot{\hat{b}}_{a} & =-\gamma_{a}\delta\mathcal{T}\left(\hat{Q}^{-1}\odot\psi(E_{P})\odot\hat{Q}\right)
\end{cases}\label{eq:NAV_Quat_Filter2}
\end{equation}

\newpage
\bibliographystyle{IEEEtran}
\bibliography{bib_Navigation}
\end{document}